\documentclass[11pt]{article}
\usepackage{amsfonts,amsmath,amssymb,amsthm,graphicx,url}
\usepackage{color}

\usepackage{algorithm}
\usepackage{algorithmic}
\usepackage{relsize}
\usepackage[margin=1in]{geometry}

\usepackage{algorithm}
\usepackage{algorithmic}
\usepackage[T1]{fontenc}

\usepackage{multirow}

\newtheorem{theorem}{Theorem}
\newtheorem{lemma}{Lemma}

\DeclareMathOperator*{\argmax}{arg\,max}

\newcommand{\bR}{\mathbb{R}}

\newcommand{\bZ}{\mathbb{Z}}

\newcommand{\cA}{{\cal A}}

\newcommand{\cU}{{\cal U}}
\newcommand{\cV}{{\cal V}}

\title{Dynamic Fair Division Problem with General Valuations}

\author{
Bo Li$^\dagger$, 
Wenyang Li$^\ddagger$, 
Yingkai Li$^\star$, 
\\ 
boli2@cs.stonybrook.edu, Stony Brook University $^\dagger$ \\
wenl6@student.unimelb.edu.au, Melbourne University$^\ddagger$ \\
yingkaili2023@u.northwestern.edu, Northwestern University $^\star$ 
}

\begin{document}

\maketitle

\begin{abstract}
In this paper, we focus on how to dynamically allocate a divisible resource fairly among $n$ players who arrive and depart over time. 
The players may have general heterogeneous valuations over the resource.
It is known that exact envy-free and proportional allocations may not exist in the dynamic setting \cite{walsh2011online}. 
Thus, we will study to what extent we can guarantee the fairness in the dynamic setting. 
We first design two algorithms which are $O(\log n)$-proportional and $O(n)$-envy-free 
for the setting with general valuations, 
and by constructing the adversary instances such that all dynamic algorithms must be at least $\Omega(1)$-proportional and $\Omega(\frac{n}{\log n})$-envy-free, 
we show that the bounds are tight up to a logarithmic factor. 
Moreover, we introduce the setting where the players' valuations are uniform on the resource but with different demands,
which generalize the setting of \cite{friedman2015dynamic}.
We prove an $O(\log n)$ upper bound and a tight lower bound for this case. 
\end{abstract}

\section{Introduction}
Initiated by the work of \cite{steinhaus1948problem}, 
the fair division problem 
has been widely studied in the literature of 
economics, mathematics and computer science
\cite{dubins1961cut,stromquist1980cut,alon1987splitting,brams1995envy,brams1996fair,robertson1998cake,aziz2016discreten}. 
It mainly considers the problem of fairly allocating a divisible resource 
among a group of players who have different preferences over the resource.

To capture fairness, many solution concepts have been proposed, 
such as  {\em proportionality} and  {\em envy-freeness} \cite{brams1996fair,neyman1946theoreme,varian1974equity,dubins1961cut}.
An allocation is proportional if each player's valuation for his received resource is at least $\frac{1}{n}$ fraction of his valuation for the whole resource, 
where $n$ is the number of players. 
An allocation is envy-free if no player values another player's allocation more than his own \cite{steinhaus1948problem}. 
One of the difficulties in finding allocations which are proportional or envy-free is that 
the resource may not be uniformly structured, such as time and land, 
and different players may hold different valuations over the same part of the resource.

Most of the previous studies have been focused on the static fair division problem, 
which assumes that all players arrive simultaneously.
Recently, the dynamic fair division problem has been considered in  
\cite{walsh2011online,friedman2015dynamic,friedman2017controlled}. 
A real-life application of the dynamic fair division problem is to fairly allocate the resource on a server among different jobs. 
In this case, different jobs arrive and depart at different time, and each job has different values for the resource allocated to it. 
The objective here is to balance the resource allocated to different jobs. 
Unlike the static setting where an envy-free allocation is guaranteed to exist \cite{brams1995envy,su1999rental},
\cite{walsh2011online} shows that envy-free or proportional allocations may not exist when the players arrive one by one over time.
Therefore, a weak version of envy-freeness is discussed in \cite{walsh2011online}, 
where the only constraint is that each player does not envy the allocations of the players who will arrive after him. 
However, this may not be sufficient to capture fairness in reality 
because the players may still observe the allocations of the previous players and envy their allocations. 
For example, an allocation which allocates the whole resource to the first player is considered to be envy-free in \cite{walsh2011online}, 
which is not plausible. 
Thus it is natural to explore to what extent we can approximate the envy-freeness and proportionality in the dynamic setting. 
This problem has been discussed in the literature, 
but they only considered the case where the players have uniform valuations. 
In \cite{friedman2015dynamic}, there is a single resource, 
and in \cite{friedman2017controlled}, there are multiple resources and the players have different demands for each resource. 
The latter one seems a general problem, but it can still be reduced to the problem with a single divisible resource while each player has piecewise linear valuations. 
We prove the~results~for~settings~with valuations more general than \cite{friedman2015dynamic,friedman2017controlled}.


\paragraph{Our Models and Results.}
In this paper, we aim at designing algorithms for the dynamic fair division (DFD) problem
so that the maximum ``dissatisfaction'' among all players is minimized.
We allow the players to have arbitrary heterogeneous yet additive valuations, 
which is a general class of valuations considered in almost all the fair division literature \cite{procaccia2015cake}. 
In our model, different players arrive and depart at different time, 
and the decision maker needs to decide each player's allocation upon his arrival. 
We say an allocation is $\xi$-envy-free if at any time, 
each player's valuation on his allocation is at least $\frac{1}{\xi}$ fraction of his value on any other player's allocation.
An allocation is $\sigma$-proportional if at any time, 
each player's value on his own allocation is at least $\frac{1}{\sigma n}$ fraction of his value on the whole resource.
We refer to $\xi$ and~$\sigma$ as the {\em approximation ratio} of the allocation.

The algorithms proposed in this paper are {\em $\tau$-recallable algorithms} with $\tau \in \bZ^+$,
which means when a new player arrives, the decision maker can recall the resource from at most $\tau$ previous players 
and reallocate it \cite{friedman2015dynamic}. 
If reallocation is not permitted, it is easy to see that the approximation ratio is unbounded,
because for any dynamic algorithm, the adversary can always choose the valuation for the second player 
such that he is only interested in the allocation of the first player. 
In this paper, we will focus on the 1-recallable algorithms
and at the end of each section, 
we will briefly discuss the $\tau$-recallable algorithms for any constant $\tau$, since they basically share the same idea and the same bound.

Our first contribution is to design two dynamic 1-recallable algorithms for the general DFD problem which are $O(n)$-envy-free and $O(\log n)$-proportional. 
Then we prove that for any dynamic 1-recallable algorithm, 
there exists an adversary instance such that the approximation ratio of the algorithm is at least $\Omega(\frac{n}{\log n})$, 
even if all the players' valuations are {\em piecewise uniform} \cite{chen2013truth}. 
Since the $\Omega(1)$ lower bound for proportionality is already proved in \cite{friedman2015dynamic}, all our bounds are tight up to a logarithmic factor.  

Because of the strong lower bound in the general setting, no one can hope to improve the fairness 
for settings with valuations more general than the piecewise uniform valuations. 
Thus, it is natural to discover the set of valuations with better dynamic performance. 
In \cite{friedman2015dynamic}, the authors show that if the players' valuations are uniform, there exist $O(1)$-envy-free algorithms.
In this paper, we generalize their model, by allowing different players to have different demands over the resource. 
When the players have uniform valuations with different demands, an allocation is {\em fair} if
(1) it meets all the players' demands or (2) each player gets at least his demand divided by total demand \cite{ghodsi2011dominant}.
As before, in the static setting, a fair allocation always exists. 
However, in the dynamic setting, the adversary can manipulate the future players to violate the fairness.
We say an allocation is {\em $\eta$-fair} if each player gets $\frac{1}{\eta}$ fraction of his allocated resource in a fair allocation.
We design an $O(\log n)$-fair 1-recallable algorithm and we prove that the bound is tight by constructing an adversary instance such that no 1-recallable algorithm can be better than $\Omega(\log n)$-fair.

\vspace{-5pt}
\paragraph{Additional Related Work.}

The fair division problem with multiple indivisible resources is also a problem widely studied in the literature 
\cite{budish2011combinatorial,procaccia2014fair,amanatidis2015approximation,brandt2016handbook,endriss2017trends}. 
Since the envy-free allocation cannot be guaranteed in this setting, the notion of fairness is captured by different relaxed versions, 
such as envy-free up to one item, which can be found by a round robin procedure \cite{budish2011combinatorial}. 
Another solution concept adopted for indivisible resourses is the maxmin share, 
which can be approximated within a factor of $2/3$ in polynomial time 
if the number of players is a constant \cite{procaccia2014fair}. 
\cite{amanatidis2015approximation} improves the result by designing an algorithm for any number of players which runs in polynomial time and guarantees the same approximation ratio. 
When the valuations of the players are submodular, the maximin fair allocation can be approximated with a factor of 10 \cite{barman2017approximation}.

%

\section{Preliminaries}\label{sec:pre}

\paragraph{Fair Division Problem.}
The fair division problem is to divide and allocate a divisible and heterogeneous resource fairly among~$n$ players, 
where the resource is represented by the real number interval $[0,1]$ 
and the player set is represented by $N=\{1,\cdots,n\}$. 
Each player $i$ has a valuation $v_i$, mapping any subset of the resource to a nonnegative value,
representing $i$'s preference over the resource.
Formally, for any set $I \subseteq [0,1]$, player $i$'s value for taking set $I$ is $v_{i}(I) \in \bR_{+}\cup \{0\}$.
The valuation profile is denoted by $v = \{v_1, \dots, v_n\}$ 
and the set of all possible valuation profiles is denoted by $\cV$.
In this paper, besides the most general valuations, we also consider two special ones.

One widely studied valuation is {\em piecewise uniform} valuation \cite{chen2013truth}.
A player has piecewise uniform valuation if and only if he has uniform valuation over a subset of the resource (which may not be a continuous interval). 
That is, for player $i$, 
the valuation $v_i$ is piecewise uniform if and only if
there exists $I_{i} \subseteq [0,1]$ 
such that $v_{i}(I) = |I \cap I_i|$ for all $I \subseteq [0, 1]$. 
By normalizing each player's valuation to range $[0, 1]$, we have that for any $I \subseteq [0, 1]$, $v_{i}(I) = \frac{|I \cap I_i|}{|I_i|}$. 

Another valuation studied in this paper is called {\em uniform with demand}.
In this model, each player has uniform valuation over the resource 
and he only cares the size of the resources he gets rather than which part of the resource is allocated to him. 
In general, different players may have different demand over the resource, 
and when the allocated resource exceeds the player's demand, 
his valuation for the allocated resource will not increase. 
For player $i$, we denote his demand as $d_{i}\in (0,1]$, 
and his valuation $v_i$ is uniform with demand if and only if 
for any $I \subseteq [0, 1]$, $v_{i}(I)=\min\{|I|, d_i\}$.  
By normalizing each player's valuation to range $[0, 1]$, we have that for any $I \subseteq [0, 1]$, $v_{i}(I)=\min\{\frac{|I|}{d_{i}}, 1\}$. 
In our paper, we adopt the normalized form of the valuations, but all the results hold for the valuations without normalization.

An {\em allocation} is denoted by a partition of $[0, 1]$, $A=(A_{0},A_{1}, \cdots, A_{n})$
where $A_{0}$ is the unallocated resource and $A_{i}$ is the resource allocated to player $i$ for any $i \in N$.
The set of all such possible partitions of $[0,1]$ is denoted by $\cU$.
A fair division algorithm (or an allocation rule) $\cA$ is a mapping from a valuation profile to an allocation,
i.e., $\cA:\cV \to \cU$.

Two widely adopted fairness solution concept are {\em proportionality} and {\em envy-freeness}. 
An allocation $A$ is proportional if and only if each player gets at least $\frac{1}{n}$ fraction of the whole resource, 
and an allocation $A$ is envy-free if and only if each player will not envy any other player's allocation. 
In this paper, we care more about the relaxed version of them.
Formally, we say an allocation $A$ is $\sigma$-proportional if $v_{i}(A_{i})\geq \frac{1}{\sigma n}v_{i}([0,1])$ for any player $i$,
and an allocation $A$ is $\xi$-envy-free if $v_{i}(A_{i})\geq \frac{1}{\xi}v_{i}(A_{j})$ for any player $i$ and $j$.

When players' valuations are uniform with demand, 
a fairness solution concept stronger than proportionality is required here 
because for a proportional allocation $A^P$, 
it only guarantees that for each player $i$, 
$v_i(A^P_i) \geq \frac{v_i([0, 1])}{n} = \frac{1}{n}$. 
However, in the static setting there exists an allocation $A$ such that 
each player~$i$ gets $\frac{d_i}{\max\{d, 1\}}$ fraction of the whole resource, 
where $d=\sum_{j\in N}d_{j}$ is the total demand of the players, 
and player~$i$'s value is guaranteed to be $v_i(A_i) = \frac{1}{\max\{d, 1\}}$. 
Since the total demand $d$ might be significantly smaller than $n$, 
$v_i(A_i)$ might be significantly larger than $v_i(A^P_i)$. 
Therefore, a stronger notion of fairness is defined as follows for uniform with demand valuations  \cite{ghodsi2011dominant}. 
Formally, an allocation $A$ is {\em fair} 
if $v_i(A_{i}) \geq \frac{1}{\max\{d, 1\}}$, 
and $A$ is $\eta$-fair if $v_i(A_{i}) \geq \frac{1}{\eta \cdot \max\{d, 1\}}$.

\vspace{-5pt}
\paragraph{Dynamic Fair Division Problem.}
Now we extend the fair division problem to the dynamic setting, where player $i$ arrives at time $t_i$ and departs at time $t^d_i$, with $t^d_i > t_i$. 
Without loss of generality, we have $0 < t_1 \leq t_2 \leq \dots \leq t_n$, 
where $n$ is the maximum number of players that could possibly arrive and is given to the algorithm as an input. 
However, the algorithm does not know the exact number of arriving players, 
or the players’ valuations until their arrival. 
The algorithm needs to allocate the resource to each player when he arrives 
without knowing future events, including his departure time. 
In this paper, we will design our algorithms and state our results for the arrival only model, 
where $t^d_1 = t^d_2 = \dots = t^d_n > t_n$. 
However, all our lower bounds and upper bounds can be directly applied for settings where players have arbitrary departure time with the same approximation ratio. 
Generally speaking, the lower bounds holds since the arrival only model is a special case of the model with departures. 
Thus the adversary instances constructed for the arrival only model can be directly applied to the model with departures to get the same bound. 
For the upper bound, we can simply recall the resource of the departed player and leave it unallocated. 
By going through our analysis for the arrival only model, we can show that all the bounds carry through in the model with departures.


Formally, the allocation of the algorithm at time $t_i$ is denoted by $A^i=(A^{i}_{0},A^{i}_{1},\cdots,A^{i}_{i})$
and the total output from time $0$ to time $t_n$ is denoted by $A = (A^{i})_{i \in \{0, 1, \cdots, n \}}$ where $A^{0} = \{A^{0}_{0}\}$ and $A^{0}_{0}=[0,1]$.
An algorithm is called $\tau$-recallable if for any $i \in [n-1]$, 
there exists $s \leq \min\{\tau, i\}$ and a set of players $S=\{i_{1},\cdots, i_{s}\}\subseteq [i]$
such that $A^{i+1}_{j} = A^{i}_{j}$ for all $j \in [i] \backslash S$, 
and $A^{i+1}_{j} \subseteq A^{i}_{j}$ for all $j\in S$.

In the following we extend the definition for fairness to the dynamic setting. 
An allocation $A$ is $\sigma$-proportional if 
$$\forall i, \forall j \leq i, v_{j}(A^{i}_j) 
\geq \frac{1}{\sigma \cdot i} v_{j}([0,1]);$$
and an allocation $A$ is $\xi$-envy-free if 
$$\forall i, \forall j,j' \leq i, v_{j}(A^i_j) 
\geq \frac{1}{\xi} v_{j}(A^i_{j'}).$$
For the uniform with demand valuations in the dynamic setting, 
an allocation $A$ is $\eta$-fair if 
$$\forall i, \forall j \leq i, 
v_j(A^i_{j}) \geq \frac{1}{\eta \cdot \max\{\sum_{l \leq i} d_l, 1\}}.$$

In this paper, we say an algorithm $\cA$ is $\sigma$-proportional, $\xi$-envy-free or $\eta$-fair for valuation space $\cV$ 
if for any valuation profile $v\in \cV$, $\cA(v)$ is $\sigma$-proportional, $\xi$-envy-free or $\eta$-fair. 
Note that, since we do not assume that the resource must be allocated to any player,
a trivial envy-free algorithm is that everyone gets nothing.
Such an algorithm is called {\em empty} in this paper.

\paragraph{Remark.} In this paper, we only consider deterministic algorithms 
because there exists a trivial randomized algorithm which is proportional and envy-free in expectation in the dynamic setting \cite{chen2013truth}. 
This randomized algorithm is also fair for uniform with demand valuations.

\section{Dynamic Fair Division (DFD) Problem}
\label{sec:uniform}

First we state the results for proportionality. 
In Theorem 6.1 of \cite{friedman2015dynamic}, no $1$-recallable algorithm can be better than $(2\ln 2)$-proportional even if 
the players' valuations are uniform, i.e., $v_i(I) = |I|, \forall i\in [n]$. 
In the following, we design the $1$-recallable algorithm $\cA_{DFD}^1$, 
defined in Algorithm \ref{alg:dfdpu:proportional}, 
which is $O(\log n)$-proportional for the DFD problem with general valuations. 
Roughly speaking, when player $i$ arrives, 
algorithm $\cA_{DFD}^1$ divides the previous players' allocated resource into $2i (3+\ln i)$ subsets
and lets player~$i$ choose his favorite bundle of resource from one player, 
where the size of the bundle depends on the current approximation ratio of any player~$j$ who arrived before player $i$. 
The performance of algorithm $\cA_{DFD}^1$ is formally analyzed in Theorem \ref{thm:prop}. 

\begin{algorithm}[htbp]
  \caption{\hspace{-3pt} 1-Recallable Algorithm $\cA_{DFD}^1$}
 \label{alg:dfdpu:proportional}
  \begin{algorithmic}[1]
\REQUIRE  A sequence of players $N=\{1,\cdots,n\}$ arriving and departing along time. 
\STATE Initially, $A^{i}_{j}=\emptyset$ for all $1\leq i \leq n$ and $0\leq j\leq i$. 
\STATE When the first player arrives, setting $A^1_1 = [0,1]$. 
\FOR{any arriving player $i > 1$}

\STATE 
		$\sigma = \lfloor 2i (3+\ln i) \rfloor$. 

\FOR{$0 < j < i$}
\STATE Partition $A^{i-1}_j$ into $\sigma$ sets $\{A^{i-1}_{j,1},\cdots,A^{i-1}_{j, \sigma} \}$ 
with $v_{j}(A^{i-1}_{j,1})=\cdots=v_{j}(A^{i-1}_{j, \sigma})$.
\STATE $\sigma_j = \lceil \frac{v_j([0,1])}{v_j(A^{i-1}_j)} \rceil$.
\ENDFOR

\STATE $(j^{*},S^{*}) =
\arg\max\limits_{
	\substack{{j < i, S\subseteq [\sigma],} \\ {|S| = \sigma -\sigma_{j}}}
}
\{ \sum _{k\in S}v_{i}(A^{i-1}_{j,k})\} $.

\STATE $A^i_i = \cup_{k\in S^{*}}A^{i-1}_{j^{*},k}$.  
\STATE $A^i_{j^{*}} = \cup_{k\notin S^{*}} A^{i-1}_{j^{*},k}$. \label{step:dfd:proportional:1}
\STATE $A^i_j = A^{i-1}_j, \forall j \in [i-1] - \{j^{*}\}$. 

\ENDFOR

\ENSURE Allocation $A=(A^{i}_{j})_{0 \leq i \leq n, 0 \leq j \leq i}$.
\end{algorithmic}
\end{algorithm}

\begin{theorem}\label{thm:prop}
For any $n \geq 1$, 1-recallable algorithm $\cA_{DFD}^1$ is $O(\log n)$-proportional for the DFD problem with $n$ players. 
\end{theorem}

\begin{proof}

In the following, we prove the theorem by induction.
When there is only $1$ player, algorithm $\cA_{DFD}^1$ is proportional since the player gets the whole resource. 
Assuming that when there are $i-1$ players, algorithm $\cA_{DFD}^1$ is $2 (3+\ln (i-1))$-proportional, 
we show that when player $i$ arrives at time $t_i$, algorithm $\cA_{DFD}^1$ is $2 (3+\ln i)$-proportional. 

First, we show that for any $j < i$, player $j$ is $2 (3+\ln i)$-proportional at time $t_i$.
Note that at least $\sigma_{j}$ elements from $\{A^{i-1}_{j,1},\cdots,A^{i-1}_{j, \lfloor 2i (3+\ln i) \rfloor} \}$ is allocated to player $j$
and player $j$ has the same value for them. 
Thus,
\begin{equation}
\frac{v_j(A^{i}_j)}{v_j([0,1])} 
\geq \frac{\sigma_j \cdot v_j(A^{i-1}_j)}{\lfloor 2i (3+\ln i) \rfloor \cdot v_j([0,1])}
\geq \frac{1}{2i (3+\ln i)}.
\end{equation}

Next, note that for player $i$, the worst case happens when his value for the resource allocated to player $j$ is 
$\frac{\rho_j \cdot v_i([0,1])}{\sum_{j \leq i-1} \rho_j}$, 
where $\rho_j = \frac{\lfloor 2i (3+\ln i) \rfloor}{\lfloor 2i (3+\ln i) \rfloor - \sigma_j}$. 
In any other cases, by choosing his favorite bundle, player $i$'s value will increase. 
Thus, 
\begin{eqnarray*}
&& \frac{v_i([0,1])}{v_i(A^{i}_i)} 
\leq \sum_{1 \leq j \leq i-1} 
\frac{\lfloor 2i (3+\ln i) \rfloor}{\lfloor 2i (3+\ln i) \rfloor - \sigma_j} \\
&\leq& 2 \sum_{1 \leq x \leq \lceil \frac{i-1}{2} \rceil} 
\frac{\lfloor 2i (3 + \ln i) \rfloor}{\lfloor 2i (3 + \ln i) \rfloor - \lceil 2(i-x) (3 + \ln (i-x)) \rceil} \\
&\leq& 2 \sum_{1 \leq x \leq \lceil \frac{i-1}{2} \rceil} 
\frac{2i (3 + \ln i)}{2i (3 + \ln i) - 2(i-x) (3 + \ln i) - 2} \\
&=& 2 i \sum_{1 \leq x \leq \lceil \frac{i-1}{2} \rceil} 
\frac{1}{x - \frac{2}{3 + \ln i}} 
\leq 2 i \left (3 + \sum_{2 \leq x \leq \lceil \frac{i-1}{2} \rceil} \frac{1}{x - \frac{2}{3 + \ln i}} \right) \\
&<& 2 i \left (3 + \sum_{1 \leq x \leq \lceil \frac{i-1}{2} \rceil - 1} \frac{1}{x} \right) 
< 2i (3 + \ln i)
\end{eqnarray*}
The second inequality holds by carefully analyzing the value $\sigma_j$ for each player $j$ who arrived before player $i$. 
Note that at any time $t_j$, the algorithm can only make the assignment 
such that the approximation ratio of the new arriving player and the player whose resource is recalled 
is $2j (3+\ln j)$, 
while the approximation ratio for all other players remain unchanged. 
Thus when player $i$ arrives, 
there exist at least $i+1-2x$ players with approximation ratio less than $2(i-x) (3+\ln (i - x))$. 
By properly eliminating the rounding in the term, the third inequality holds. 
The fourth inequality holds because the first term in the summation is at most 3. 
The last two inequalities hold trivially. 
Therefore, by induction, algorithm $\cA_{DFD}^1$ is $O(\log n)$-proportional and Theorem \ref{thm:prop} holds.
\end{proof}

Next we show how to design the $O(n)$-envy-free algorithm $\cA_{DFD}^2$, which is quite similar to algorithm $\cA_{DFD}^1$. 
Briefly speaking, $\cA_{DFD}^2$ allocates the whole resource to the first player. 
Then, when player $i$ arrives, 
for each player $j<i$, $\cA_{DFD}^2$ uniformly divides the resource allocated to player $j$ into $i$ equal subsets 
and let player~$i$ chooses his favorite set. 
The formal description is omitted here due to the space limit. 

%
%
%
%
%
%
%

\begin{theorem}\label{thm:envy:up}
For any $n \geq 1$, 1-recallable algorithm $\cA_{DFD}^2$ is $O(n)$-envy-free for the DFD problem with $n$ players. 
\end{theorem}
\vspace{-5pt}

\begin{proof}
Again we prove by induction.
When player 1 arrives, $A^1_1$ is exactly $[0, 1]$.
That is, $\cA_{DFD}^2$ is envy-free. 
Assume that for player $i \geq 2$, $\cA_{DFD}^2$ is $(i-1)$-envy-free at time $t_{i-1}$. 
Now let us prove that all $i$ players are $i$-envy-free at time $t_i$ 
by analyzing their value at time $t_i$. 

First, we prove that player $i$ is $i$-envy-free at time $t_i$. 
Assume that the allocation of player $j$, $A^{i-1}_j$, is partitioned into~$i$ sets $\{A^{i-1}_{j,1},\cdots,A^{i-1}_{j,i} \}$. 
Since player $i$ is allocated with his favorite set $A^{i-1}_{j^{*},k^{*}} \in \{A^{i-1}_{j,k}\}_{j<i,k\leq i}$, 
for any $j < i$,
\begin{equation}\label{eq:dfd:up:envy:1}
v_i(A^i_i)\geq v_{i}(A^{i-1}_{j^{*},k^{*}}) \geq \frac{1}{i} v_i(A^{i-1}_j) \geq \frac{1}{i} v_i(A^{i}_j).
\end{equation}

Next, we show that for any $j \in [i-1]$, player $j$ is $i$-envy-free at time $t_i$.
By induction hypothesis, for any $j, j' \in [i-1]$, 
$$v_j(A^{i-1}_j) \geq \frac{1}{i-1} v_j(A^{i-1}_{j'}).$$ 
Note that $A^{i-1}_j$ is partitioned into $i$ sets with equal values, 
and at most one of them is allocated to player $i$.
Hence, for any $j, j' \in [i-1]$,
\begin{equation}\label{eq:dfd:up:envy:2}
v_j(A^{i}_j) \geq \frac{i-1}{i} v_j(A^{i-1}_j) 
\geq \frac{1}{i} v_j(A^{i-1}_{j'})
\geq \frac{1}{i} v_j(A^{i}_{j'}).
\end{equation} 

Since the allocation of player $i$ at time $t_i$ is also a subset of player $j^*$'s allocation at time $i-1$, 
according to Inequality \ref{eq:dfd:up:envy:2}, we have 
\begin{equation}\label{eq:dfd:up:envy:3}
v_j(A^{i}_j) \geq \frac{1}{i} v_j(A^{i-1}_{j^{*}}) \geq \frac{1}{i} v_j(A^{i}_i).
\end{equation}

Combining Inequalities \ref{eq:dfd:up:envy:1}, \ref{eq:dfd:up:envy:2} and \ref{eq:dfd:up:envy:3}, 
all players are $i$-envy-free at time $t_i$. 
By induction, algorithm $\cA_{DFD}^2$ is $n$-envy-free.  
Thus finishes the proof of Theorem \ref{thm:envy:up}.
\end{proof}

We have designed an $1$-recallable algorithm which is $O(n)$-envy-free for the DFD problem with general valuations. 
Next we show that this bound is almost tight by proving that 
all $1$-recallable algorithms must be at least $\Omega(\frac{n}{\log n})$-envy-free, 
even if the players' valuations are piecewise uniform (DFDPU). 
Formally, we have the following theorem. 


\begin{theorem}\label{thm:envy:low}
For any 1-recallable algorithm $\cA$ (except empty algorithm) for the DFDPU problem, 
$\cA$ is $\Omega(\frac{n}{\log n})$-envy-free with $n$ players.
\end{theorem}
\begin{proof}
Assume there exists an algorithm $\cA$ which is $\xi$-envy-free and not empty. 
Consider the following adversary instance where for player $1$, 
$v_1(I) = |I|$ 
and for any player $i \geq 2$, 
$v_i(I) = \frac{|I \cap A^{i-1}_1|}{|A^{i-1}_1|}$. 
That is, the first player wants the whole resource 
and each new arriving player only wants the resource allocated to the first player by $\cA$ at time $t_{i-1}$. 
Since $\cA$ is $\xi$-envy-free and not empty, we must have $A^i_1 \neq \emptyset, \forall i \leq n$, 
and the adversary instance is well defined. 

First we consider the envy-freeness for player $i \geq 2$.
Since the resource that player $i$ values is exactly the resource allocated to player $1$ at time $t_{i-1}$,
to ensure that player $i$ is $\xi$-envy-free at time $t_i$, 
the algorithm must recall the resource from player $1$'s 
and reallocate at least $\frac{1}{1+\xi}$ fraction of the resource to player $i$. 
Therefore, 
$$v_1(A^{i}_1) \leq (1-\frac{1}{1+\xi}) v_1(A^{i-1}_1),$$
for any $i \geq 2$. 
Accordingly, 
\begin{equation}\label{eq:dfdpu:lb:1}
v_1(A^{n}_1) \leq (1- \frac{1}{1+\xi})^{n-2} v_1(A^{2}_1).
\end{equation}

When all $n$ players have arrived, to ensure that player $1$ is $\xi$-envy-free, we have 
\begin{eqnarray}
&& v_1(A^{n}_1) \geq \frac{1}{\xi} v_1(A^{n}_2) 
= \frac{1}{\xi} v_1(A^2_2) \nonumber \\
&=& \frac{1}{\xi |A^1_1|} v_2(A^2_2)  
\geq \frac{1}{\xi^2 |A^1_1|} v_2(A^2_1) = \frac{1}{\xi^2} v_1(A^2_1).\label{eq:dfdpu:lb:2}
\end{eqnarray} 

The first and the second inequality holds due to the $\xi$-envy-freeness of player $1$ at time $t_n$ 
and player $2$ at time $t_2$. 
The first equality holds because player $1$ is always the player who needs to be recalled. 
Therefore, 
$A^{i}_2 = A^2_2$ for any $2\leq i\leq n$.
The second and the last equality holds because by the construction of the valuations of player $1$ and player $2$, 
$v_1(I) = \frac{1}{|A^1_1|} v_2(I)$ for any $I \subseteq A^1_1$. 

Combining inequalities \ref{eq:dfdpu:lb:1} and \ref{eq:dfdpu:lb:2}, 
we have
$$(1-\frac{1}{1+\xi})^{n-2} \geq \frac{1}{\xi^2}.$$ 
By solving the above inequality,
$\xi = \Omega(\frac{n}{\log n})$ and Theorem \ref{thm:envy:low} holds.
\end{proof}


\paragraph{Remark.}
Given a valuation profile $v$, we say an allocation $A$ is {\em non-wasteful} if for all player $i$, 
we only allocate the resource that he values to him.
That is, $\forall i \leq j \leq n$, $\forall I \subseteq A^j_i$, $v_i(I) > 0$.
We say an algorithm $\cA$ is {\em non-wasteful} for valuation space $\cV$ if for any $v\in \cV$, $\cA(v)$ is non-wasteful.
Then if we consider the $\tau$-recallable algorithm where $\tau$ can be any constant, 
the adversary instance in the proof of Theorem \ref{thm:envy:low} still shows that 
no non-wasteful algorithm can be $o(\frac{n}{\log n})$-envy-free even for piecewise uniform valuations.
 


\section{When the Valuation Function is Uniform with Demand (UD)}
\label{sec:demand}
In this section, we study the dynamic fair division problem with uniform with demand valuations. 
Since in this model, the players do not care which part of the resource is allocated to him, 
we regard the allocation $A^{i}_{j}$ in $A=(A^{i}_{j})_{0\leq i\leq n, 0\leq j\leq i}$ as the size of the allocated resource.

For the UD problem, a simple greedy algorithm is $2$-envy-free and $2$-proportional. 
However, our main focus in this section is its performance under the more demanding fairness solution concept. 
First we consider the following special case where the smallest demand is within a constant fraction of the largest demand.
This case will provide enough intuition about solving the general case. 
In order to make our algorithm clear, 
we explicitly write out all the parameters in $\cA^{S}_{UD} (d,c,\eta)$, 
where the algorithm has those extra parameters as input. 
Here parameter $d$ is the minimum demand in the adversary instance 
and $c$ is the ratio between the largest and smallest demand. 
Note that $cd \leq 1$. 
Also, $\eta$ is an extra parameter representing the total amount of resource can be allocated to the arriving players. 
Algorithm $\cA^{S}_{UD}$ is formally defined in Algorithm \ref{alg:ud:lem}. 
The general idea of the algorithm is to treat each arriving player's demand as the maximum demand, 
and the total demand as the minimum sum. 
Then the algorithm only allocates the resource to each arriving player such that the approximation ratio is not violated. 
The performance of the algorithm is analyzed in Lemma \ref{lem:bd}.

\begin{algorithm}[htbp]
  \caption{\hspace{-3pt} 1-Recallable Algorithm $\cA^{S}_{UD} (d,c,\eta)$}
 \label{alg:ud:lem}
  \begin{algorithmic}[1]
\REQUIRE A sequence of players $N=\{1,\cdots,n\}$ arriving and departing along time. 

\STATE Initially, $A^{i}_{j}=0$ for all $0\leq i \leq n$, $0\leq j\leq i$. 

\FOR{any arriving player $i \geq 1$}
\STATE Let $j^{*} = \arg\max_{1\leq j \leq i-1} \{A_{j}^{i-1}\}$.

\STATE Set $A_{i}^{i} = A_{j^{*}}^{i} = \frac{d}{2\eta \ln 3\cdot \max\{d\cdot i, 1\}}$. \label{step:ud:lem:1}

\STATE $A^i_j = A^{i-1}_j, \forall j \in [i-1] - \{j^{*}\}$. 
\STATE $A_{0}^{i} = A^{0}_{0}- \sum_{j\in [i]} A^{i}_{j}$. 
\ENDFOR

\ENSURE Allocation $A=(A^{i}_{j})_{0 \leq i \leq n, 0 \leq j \leq i}$.
\end{algorithmic}
\end{algorithm}

\begin{lemma}\label{lem:bd} 
$\forall c, d, \eta > 0$, 
algorithm $\cA^{S}_{UD}$ is $(2c \eta \ln 3)$-fair and at most $\frac{1}{\eta}$ fraction of the resource is allocated.
\end{lemma}
\begin{proof}

We first show that algorithm $\cA^{S}_{UD}$ is well defined 
by showing that the reallocation in Step \ref{step:ud:lem:1} is always feasible
and the total allocated resource is no more than $\frac{1}{\eta}$. 
First note that $\frac{d}{2\eta \ln 3 \cdot \max\{i \cdot d, 1\}}$ is non-increasing with respect to $i$. 
Next we divide the analysis into three cases. 

Case 1: $i \leq \lfloor \frac{1}{d} \rfloor$.
When $i \leq \lfloor \frac{1}{d} \rfloor$, $A_{j}^{i} = \frac{d}{2 \eta \ln 3}$ for all $j \leq i$.
Therefore, $\sum_{j\in [i]} A_{j}^{i} = \frac{id}{2\eta \ln 3} \leq \frac{1}{2\eta \ln 3} < \frac{1}{\eta}$.
In this case, no resource is recalled and the total allocation never exceed~$\frac{1}{\eta}$. 

Case 2: $\lfloor \frac{1}{d} \rfloor < i \leq 2\lfloor \frac{1}{d} \rfloor$. 
In this case, the first $\lfloor \frac{1}{d} \rfloor$ players still satisfy Case 1, 
but when the $(\lfloor \frac{1}{d} \rfloor +l)$-th player arrives, player $l$'s allocation will be recalled, 
where $1\leq l \leq i-\lfloor \frac{1}{d} \rfloor$.
Then both player $l$ and $(\lfloor \frac{1}{d} \rfloor +l)$ will be allocated with 
$\frac{1}{2\eta \ln 3 \cdot (\lfloor \frac{1}{d} \rfloor + l)}$.
The total allocated resource is 
$$
( 2\lfloor \frac{1}{d} \rfloor - i)\frac{d}{2\eta \ln 3}  +
2 \sum_{l =1 }^{i-\lfloor \frac{1}{d} \rfloor} \frac{1}{2\eta \ln 3 \cdot (\lfloor \frac{1}{d} \rfloor + l)} \leq \frac{1}{\eta}.
$$ 
The above inequality holds by analyzing the monotonicity.

Case 3: $i > 2\lfloor \frac{1}{d} \rfloor$.
When $i > 2\lfloor \frac{1}{d} \rfloor$, all the resource allocated to the first $\lfloor \frac{i}{2}\rfloor$ players will be recalled.
In Algorithm $\cA^{S}_{UD}$, 
whenever the resource is recalled from a player, 
the algorithm will create two sets of resource with the same size, 
and allocate them to both the recalled player and the new arriving player. 
Thus the total allocated resource in this case is at most twice of the resource allocated to players from 
$\lfloor \frac{i}{2}\rfloor$ to $i$, that is, 
%
%
%
%
%
\begin{eqnarray*}
2 \sum_{l = \lfloor \frac{i}{2} \rfloor}^{i} \frac{1}{2 \eta \ln 3 \cdot l} 
\leq \frac{1}{\eta \ln 3} \ln \frac{i}{\lfloor \frac{i}{2} \rfloor} \leq \frac{1}{\eta}.
\end{eqnarray*}

Finally, we analyze the performance of $\cA^{S}_{UD}$. 
At anytime $t_i$, $1 \leq i \leq n$, 
since player $i$'s demand $d_i \leq c \cdot d$, 
his value for allocation $A^i_i$ is 
$v_i(A^i_i) \geq \frac{1}{2c\eta\ln 3 \cdot \max\{i \cdot d, 1\}}$. 
When player $i$ arrives, the total demand $\sum_{l \leq i} d_l \geq d \cdot i$, 
and the value of player $i$ at time $t_i$ is 
$$
v_i(A^i_i) \geq 
\frac{1}{2c\eta\ln 3 \cdot \max\{\sum_{l \leq i} d_l, 1\}}.
$$
Similarly, the value of any player $j < i$ satisfies the above inquality. 
Therefore, algorithm $\cA^{S}_{UD}$ is $(2c \eta \ln 3)$-fair, 
and thus finishes the proof of Lemma \ref{lem:bd}.
\end{proof}

Next we design the algorithm $\cA_{UD}$ for the valuations with arbitrary demands, 
which is formally defined in Algorithm~\ref{alg:demand}. 
The idea for this algorithm is to 
classify the players with demands that differ up to a constant factor into the same group 
and run the algorithm $\cA^S_{UD}$ for the players in each group. 
By analyzing the performance of algorithm $\cA_{UD}$, we have the following theorem. 

\begin{algorithm}[htbp]
  \caption{\hspace{-3pt} 1-Recallable Algorithm $\cA_{UD}$}
 \label{alg:demand}
  \begin{algorithmic}[1]
\REQUIRE A sequence of players $N=\{1,\cdots,n\}$ arriving and departing along time. 
Each player $i\in N$ has demand $d_{i}$. 
Let $m = \lceil \log n \rceil$.

\STATE Divide the demand space $[0,1]$ into $m +1$ sets 
	$(S_0, S_1, \dots, S_m)$, where 
	$S_{0} = [0, 2^{-m}]$, 
	$S_l = (2^{-l}, 2^{1-l}], l \in [m]$.

\FOR{player $i$ with demand $d_i$} 
	\STATE Find $j$ such that $d_i \in S_j$. 
	\STATE $A^i_i = \frac{2^{1-j}}{2\eta \ln 3 \max\{ \sum_{l \leq i} d_l, 1\}}$. 
	\STATE Let $i^* = \argmax_{i'} \{A^{i-1}_{i'} | d_{i'} \in S_j\}$. 
	\STATE $A^i_{i^*} = A^i_i$, 
		and $A^i_{i'} = A^{i-1}_{i'}, \forall i'\neq i, i^*$.
	\STATE $A^i_0 = 1-\sum_{j\in [i]} A^i_j$. 
\ENDFOR 

\ENSURE Allocation $(A^{i}_{j})_{0 \leq i \leq n, 0 \leq j \leq i}$.
\end{algorithmic}
\end{algorithm}

\begin{theorem}\label{thm:demand}
$\forall n \geq 1$, 1-recallable algorithm $\cA_{UD}$ is $O(\log n)$-fair for the UD problem with $n$ players. 
\end{theorem}

\begin{proof}
It is clear that $\cA_{UD}$ classifies all players in $m+1$ sets. 
Let $N_j$ denote the set of players with demand in $S_j$. 
Then except the players in $N_{0}$, the demands of the players in the same set differ up to a constant factor $2$. 

First we find the parameter $\eta$ such that algorithm $\cA_{UD}$ is feasible. 
For players in $N_{0}$, since the maximum demand is at most $\frac{1}{n}$ and there are at most $n$ players in total, 
the resource allocated to those players is bounded by $\frac{1}{2\eta \ln 3}$. 
For player $i$ in $N_j$, where $j\in [m]$, 
the total demand when player $i$ arrives is at least the total demand of the players in $N_j$ who arrived before player $i$. 
Using the similar argument in Lemma \ref{lem:bd} and applying the case studies, 
the total resource allocated to players in $N_j$ is at most $\frac{1}{\eta}$. 
By setting $\eta = 1+m = 1 + \lceil \log n \rceil$, algorithm $\cA_{UD}$ is always feasible. 
Then, similar to Lemma \ref{lem:bd}, 
it is easy to get that algorithm $\cA_{UD}$ is $(4\eta \ln 3)$-fair, 
which is $O(\log n)$-fair.
Therefore, Theorem \ref{thm:demand} holds. 
\end{proof}

%

Next we construct an adversary instance for the UD problem to show that the bound of our algorithm is tight. 
Formally, we have the following theorem.  

\begin{theorem}\label{thm:low:ud}
For any 1-recallable algorithm $\cA$ for the UD problem, $\cA$ is 
$\Omega(\log n)$-fair with $n$ players.
\end{theorem}
\begin{proof}
Let $\cA$ be any 1-recallable algorithm which is $\eta$-fair. 
Here we consider the following adversary instance.
In order to make the instance clear, we describe the arriving players by $m=\frac{\log n}{3}$ rounds 
($n$ is a large enough integer such that $m$ is also a large enough integer).
Formally, within stage $j \leq m$, $n_{j} = \frac{n}{2\times 4^{j-1}}$ players arrive one by one and each of them has demand $\frac{8^j}{n}$.
That is, in the first round, $\frac{n}{2}$ players arrive one by one and each of them has demand $\frac{8}{n}$.
For each of the following rounds, the number of the arriving players decreases by a multiplicative factor of $4$, 
but the demand of each player increases by a multiplicative factor of $8$. 
Note that in the last round the demand of each player is $\frac{8^m}{n}=1$. 
Also, the total number of players in the designed instance is 
$\sum_{j\leq m} n_j < n$.
Thus this instance is well defined, 
and it is sufficient to prove that all 1-recallable algorithm is $\Omega(\log n)$-fair for this instance. 

In each stage $j\leq m$, the total number of players that would arrive in the future is less than 
$$u = (\frac{n}{2\times 4^j}) \big/ ({1-\frac{1}{4}}) = \frac{n}{6\times 4^{j-1}}.$$
It is easy to see that $u=\frac{1}{3}n_{j}$.
Therefore, at least $\frac{2}{3}n_{j}=\frac{n}{3\times 4^{j-1}}$ players 
who arrive during the $j$th round won't get recalled by in the future. 
Then we can lower bound the resource allocated to those players. 
First note that the total demand before stage $j$ is 
$\sum_{j' < j} \frac{n}{2\times 4^{j'-1}} \times \frac{8^{j'}}{n} 
= \sum_{j' < j} \frac{8^{j'}}{2\times 4^{j'-1}}$. 
The total demand when the $i$th player in stage $j$ arrives is 
$$
D_i = \sum_{j' < j} \frac{8^{j'}}{2\times 4^{j'-1}} 
+ \frac{8^j \cdot i}{n}.
$$

Since algorithm $\cA$ is an 1-recallable algorithm, 
for any integer $x \leq \frac{n_j + 1}{2}$, 
there exists at least $n_j+1-2x$ players whose allocated resource is at least $\frac{8^j}{n\eta \cdot D_{n_j-x}}$. 
Thus the total resource allocated to the players in stage $j$ is 
\begin{eqnarray*}
&& L_j \geq \frac{2}{\eta} \sum_{i = 1}^{\frac{n}{6\times 4^{j-1}}} 
\frac{8^j}{n\cdot D_{n_j-i}}\\
&=& \frac{2}{\eta} \sum_{i = 1}^{\frac{n}{6\times 4^{j-1}}} 
\frac{\frac{8^j}{n}}
{ \sum_{j' < j} \frac{8^{j'}}{2\times 4^{j'-1}} 
+ \frac{8^j}{2\times 4^{j-1}} - \frac{8^j \cdot i}{n}} \\
&=& \frac{2}{\eta} \sum_{i = 1}^{\frac{n}{6\times 4^{j-1}}} 
\frac{\frac{8^j}{n}}
{2^{j+2} - 2 - \frac{8^j \cdot i}{n}} \\
&\geq& \frac{2}{\eta} \sum_{i = 1}^{\frac{n}{6\times 4^{j-1}}} 
\frac{\frac{8^j}{n}}
{2^{j+2} - \frac{8^j \cdot i}{n}} 
= \frac{2}{\eta} \sum_{i = 1}^{\frac{n}{6\times 4^{j-1}}} 
\frac{1}
{\frac{4n}{4^j} - i} \\
&\to& \frac{2}{\eta} \left[\ln \frac{4n}{4^j} - \ln (\frac{4n}{4^j} - \frac{n}{6\times 4^{j-1}}) \right]
= \frac{2 \ln \frac{6}{5}}{\eta},
\end{eqnarray*} 
when we set $n \to \infty$. 
Thus the total allocated resource is 
$$
1 \geq \sum_{j \leq \frac{\log n}{3}} L_j \geq \frac{2 \ln \frac{6}{5} \log n}{3\eta}.
$$
Therefore, $\eta \geq \frac{2}{3} \ln \frac{6}{5} \log n$ and Theorem \ref{thm:low:ud} holds.
\end{proof}

\paragraph{Remark} If we consider $\tau$-recallable algorithms here, where $\tau$ is a constant, 
by designing the similar adversary instance, 
except that the number of players and the demand change in the ratio of $\frac{1}{4\tau}$ and $8\tau$ respectively between each stage, 
we can still obtain the $\Omega(\log n)$ lower bound.

\section{Conclusion and Future Direction}

In this paper, we study the dynamic fair division problem with a divisible resource. 
We design $O(\log n)$-proportional and $O(n)$-envy-free algorithms for players with general valuations.
We also show that those bounds are tight up to a logarithmic factor.
However, for the valuations that are uniform with demand, 
we design an $O(\log n)$-fair algorithm and we prove that this ratio is tight.

There are a lot of future directions that are worth exploring. 
An immediate open question is whether we can find algorithms with better approximation ratio for $\tau$-recallable algorithms without the non-wasteful assumption when $\tau \geq 2$.  
Another particularly interesting question is to consider the problem via a game theoretic view in the dynamic setting.
Moreover, the dynamic version of allocating multiple indivisible resources is also worthy of the effort. 

%

\bibliographystyle{abbrv}
\bibliography{ref}

\begin{thebibliography}{10}

\bibitem{alon1987splitting}
N.~Alon.
\newblock Splitting necklaces.
\newblock {\em Advances in Mathematics}, 63(3):247--253, 1987.

\bibitem{amanatidis2015approximation}
G.~Amanatidis, E.~Markakis, A.~Nikzad, and A.~Saberi.
\newblock Approximation algorithms for computing maximin share allocations.
\newblock In {\em International Colloquium on Automata, Languages, and
  Programming (ICALP'15)}, pages 39--51. Springer, 2015.

\bibitem{aziz2016discreten}
H.~Aziz and S.~Mackenzie.
\newblock A discrete and bounded envy-free cake cutting protocol for any number
  of agents.
\newblock In {\em 57th Annual Symposium on Foundations of Computer Science
  (FOCS'16)}, pages 416--427, 2016.

\bibitem{barman2017approximation}
S.~Barman and S.~K.~K. Murthy.
\newblock Approximation algorithms for maximin fair division.
\newblock pages 647--664, 2017.

\bibitem{brams1995envy}
S.~J. Brams and A.~D. Taylor.
\newblock An envy-free cake division protocol.
\newblock {\em The American Mathematical Monthly}, 102(1):9--18, 1995.

\bibitem{brams1996fair}
S.~J. Brams and A.~D. Taylor.
\newblock {\em Fair Division: From cake-cutting to dispute resolution}.
\newblock Cambridge University Press, 1996.

\bibitem{brandt2016handbook}
F.~Brandt, V.~Conitzer, U.~Endriss, A.~D. Procaccia, and J.~Lang.
\newblock {\em Handbook of computational social choice}.
\newblock Cambridge University Press, 2016.

\bibitem{budish2011combinatorial}
E.~Budish.
\newblock The combinatorial assignment problem: Approximate competitive
  equilibrium from equal incomes.
\newblock {\em Journal of Political Economy}, 119(6):1061--1103, 2011.

\bibitem{chen2013truth}
Y.~Chen, J.~K. Lai, D.~C. Parkes, and A.~D. Procaccia.
\newblock Truth, justice, and cake cutting.
\newblock {\em Games and Economic Behavior}, 77(1):284--297, 2013.

\bibitem{dubins1961cut}
L.~E. Dubins and E.~H. Spanier.
\newblock How to cut a cake fairly.
\newblock {\em The American Mathematical Monthly}, 68(1):1--17, 1961.

\bibitem{endriss2017trends}
U.~Endriss.
\newblock {\em Trends in Computational Social Choice}.
\newblock Lulu. com, 2017.

\bibitem{friedman2015dynamic}
E.~Friedman, C.-A. Psomas, and S.~Vardi.
\newblock Dynamic fair division with minimal disruptions.
\newblock In {\em 16th ACM Conference on Economics and Computation (EC'15)},
  pages 697--713, 2015.

\bibitem{friedman2017controlled}
E.~Friedman, C.-A. Psomas, and S.~Vardi.
\newblock Controlled dynamic fair division.
\newblock In {\em 18th ACM Conference on Economics and Computation (EC'17)},
  pages 461--478, 2017.

\bibitem{ghodsi2011dominant}
A.~Ghodsi, M.~Zaharia, B.~Hindman, A.~Konwinski, S.~Shenker, and I.~Stoica.
\newblock Dominant resource fairness: Fair allocation of multiple resource
  types.
\newblock In {\em NSDI}, volume~11, pages 24--24, 2011.

\bibitem{neyman1946theoreme}
J.~Neyman.
\newblock Un theoreme d'existence.
\newblock 1946.

\bibitem{procaccia2015cake}
A.~D. Procaccia.
\newblock Cake cutting algorithms.
\newblock In {\em Handbook of Computational Social Choice, chapter 13}.
  Citeseer, 2015.

\bibitem{procaccia2014fair}
A.~D. Procaccia and J.~Wang.
\newblock Fair enough: Guaranteeing approximate maximin shares.
\newblock In {\em Proceedings of the fifteenth ACM conference on Economics and
  computation}, pages 675--692. ACM, 2014.

\bibitem{robertson1998cake}
J.~Robertson and W.~Webb.
\newblock Cake-cutting algorithms: Be fair if you can.
\newblock 1998.

\bibitem{steinhaus1948problem}
H.~Steinhaus.
\newblock The problem of fair division.
\newblock {\em Econometrica}, 16:101--104, 1948.

\bibitem{stromquist1980cut}
W.~Stromquist.
\newblock How to cut a cake fairly.
\newblock {\em The American Mathematical Monthly}, 87(8):640--644, 1980.

\bibitem{su1999rental}
F.~E. Su.
\newblock Rental harmony: Sperner's lemma in fair division.
\newblock {\em The American mathematical monthly}, 106(10):930--942, 1999.

\bibitem{varian1974equity}
H.~R. Varian.
\newblock Equity, envy, and efficiency.
\newblock {\em Journal of economic theory}, 9(1):63--91, 1974.

\bibitem{walsh2011online}
T.~Walsh.
\newblock Online cake cutting.
\newblock {\em Algorithmic Decision Theory}, pages 292--305, 2011.

\end{thebibliography}

\end{document}